\documentclass[11pt]{amsart}

\usepackage{subcaption}

\usepackage{amsmath}
\usepackage{amssymb}
\usepackage{amsthm}
\usepackage{graphicx}
\usepackage{mathtools}
\mathtoolsset{showonlyrefs}
\usepackage{url}
\usepackage{enumerate}
\usepackage{color}

\usepackage{tikz}
\usepackage{pgfplots}
\usepgfplotslibrary{fillbetween}
\pgfplotsset{compat=1.11}
\usetikzlibrary{patterns}

\usepackage[margin=1.35in]{geometry}

\def\E{\mathbb{E}}
\def\Lam{\Lambda}

\def\E{\mathbb{E}}

\def\Z{\mathbb{Z}}
\def\eps{\epsilon}
\def\del{\delta}

\def\lam{\lambda}

\def\cG{\mathcal G}
\def\cC{\mathcal C}
\def\cP{\mathcal P}

\def\cI{\mathcal I}

\def\kotecky{Koteck\'y}

\newtheorem*{theorem*}{Theorem}
\newtheorem{theorem}{Theorem}[section]
\newtheorem{lemma}[theorem]{Lemma}
\newtheorem{cor}[theorem]{Corollary}

\newtheorem*{prop*}{Proposition}

\newtheorem*{claim*}{Claim}
\newtheorem*{fact*}{Fact}
\newtheorem*{remark*}{Remark}
\newtheorem*{defn*}{Definition}

\theoremstyle{definition}

\begin{document}

\title[Counting independent sets in unbalanced bipartite graphs]{Counting independent sets in unbalanced \\bipartite graphs}

\author{Sarah Cannon}
\address{University of California Berkeley}
\email{sarah.cannon@berkeley.edu}
\author{Will Perkins}
\address{University of Illinois at Chicago}
 \email{math@willperkins.org}

\begin{abstract}
We give an FPTAS for approximating the partition function of the hard-core model for bipartite graphs when there is sufficient imbalance in the degrees or fugacities between the sides $(L,R)$ of the bipartition.  This includes, among others, the biregular case when $\lam=1$ (approximating the number of independent sets of $G$) and $\Delta_R \geq 7\Delta_L \log(\Delta_L)$.  Our approximation algorithm is based on truncating the cluster expansion of a polymer model partition function that expresses the hard-core partition function in terms of deviations from independent sets that are empty on one side of the bipartition.
As a consequence of the method, we  also prove that the hard-core model on such graphs exhibits exponential decay of correlations by utilizing connections between the cluster expansion and joint cumulants. 
\end{abstract}

\maketitle

\section{Introduction}

The computational complexity of approximating the number of independent sets in a bipartite graph is a central open problem in the field of approximate counting:  neither a general polynomial-time algorithm nor a proof of NP-hardness is known.  The problem defines a complexity class \#BIS (bipartite independent set) and many other important counting problems have been shown to be as hard to approximate as \#BIS, including counting weighted independent sets in  bounded-degree bipartite graphs when the weighting factor $\lam$ is large enough.   Counting weighted independent sets arises naturally in statistical physics as it is exactly the problem of computing the partition function of the hard-core model at \textit{fugacity} $\lam$.  

We give an FPTAS for both the weighted and unweighted counting problems, provided that the bipartite graph has bounded degree and there is a sufficient imbalance in the degrees or fugacities between the sides of the bipartition. 

More formally, the hard-core model on a graph $G$ is a probability distribution over the independent sets of $G$ given by
\begin{align*}
\mu_G(I) &= \frac{\lam^{|I|}}{Z(G,\lam)}
\end{align*}
where $\lam > 0$ is the {fugacity} and the normalizing constant, or the partition function, is
\begin{align*}
Z(G,\lam) &= \sum_{I \in \cI(G)} \lam^{|I|} \, ,
\end{align*}
where $\cI(G)$ is the set of all independent sets of $G$. 

More generally, one can consider the multivariate hard-core model, assigning a fugacity $\lam_v$ to each vertex $v \in V(G)$. The resulting partition function is
\begin{align*}
Z(G) &= \sum_{I \in \cI(G)} \prod_{v \in I} \lam_v \,.
\end{align*}
 We focus here on a bivariate hard-core model, with a fugacities $\lam_L, \lam_R$ assigned to vertices on the respective sides of a bipartite graph with bipartition $(L,R)$.  We denote the partition function by $Z(G,\lam_L,\lam_R)$ or $Z(G)$ for short. 

 There are two natural computation problems associated to spin models such as the hard-core model: the approximate counting problem and the approximate sampling problem.
An $\eps$-relative approximation to $Z(G)$ is a number $\hat Z$ so that 
\begin{align*}
e^{-\eps} \hat Z \le Z(G) \le e^{\eps} \hat Z \,.
\end{align*}
A \textit{fully polynomial-time approximation scheme} (FPTAS) for $Z(G)$ is an algorithm that given $G$ and $\eps>0$, outputs an $\eps$-relative approximation to $Z(G)$ and runs in time polynomial in $|V(G)|$ and $1/\eps$.  An \textit{efficient sampling scheme} is a randomized algorithm that outputs an independent set from $\cI(G)$ with distribution $\hat \mu$ so that $\| \mu_G - \hat \mu \|_{TV}<\eps$ and runs in time polynomial in $|V(G)|$ and $1/\eps$.
 
In general, the computational complexity of approximating $Z(G)$ is well understood. For graphs of maximum degree $\Delta$, there is an FPTAS due to Weitz~\cite{weitz2006counting} when $\lam < \lam_c(\Delta)$,  the uniqueness threshold for the infinite $\Delta$-regular tree. On the other hand, Sly~\cite{sly2010computational}, Sly and Sun~\cite{sly2014counting}, and Galanis, Stefankovic, and Vigoda~\cite{galanis2011improved} showed that there is no polynomial-time approximation algorithm for $\lam > \lam_c(\Delta)$ unless $\text{NP} = \text{RP}$.  

For bipartite graphs, however, the situation is far from clear.  There is no hardness known, and one might expect that the problem becomes easier as $\lam$ gets large since finding a maximum size independent set in a bipartite graph is tractable, unlike in general graphs. In fact, Dyer, Greenhill, Goldberg, and Jerrum~\cite{dyer2004relative} defined the complexity class \#BIS to capture the complexity of approximating the number of independent sets (i.e. approximating $Z(G)$ at $\lam=1$) in bipartite graphs.  The complexity of \#BIS is still unresolved, but many other important approximate counting problems have been shown to be \#BIS-hard, including approximating the partition function of the $q$-color ferromagentic Potts model ($q \ge 3$)~\cite{goldberg2012approximating,galanis2016ferromagnetic}.  Refined results of Cai, Galanis, Goldberg, Guo, Jerrum, Stefankovic, and Vigoda show that for bipartite graphs of maximum degree $\Delta$, it is already \#BIS-hard to approximate $Z(G)$ for any $\lam > \lam_c(\Delta)$~\cite{cai2016hardness}.

There are a handful of algorithmic results on approximating $Z(G)$ that exploit bipartite structure.  Liu and Lu~\cite{liu2015fptas} gave an FPTAS for $Z(G)$ at $\lam=1$ on bipartite graphs when the degree on one side of the bipartition is bounded by $5$, with arbitrary degrees on the other side.  Helmuth, Perkins, and Regts~\cite{helmuth2018contours} used contour models from Pirogov-Sinai theory in statistical physics to give an FTPAS for the hard-core model on subsets of $\Z^d$ for sufficiently large $\lam$.   Building on this approach, Jenssen, Keevash, and Perkins~\cite{JenssenAlgorithmsSODA} gave an FPTAS for the hard-core model on bipartite expander graphs at large $\lam$, and these results were sharpened in the case of random regular bipartite graphs~\cite{JKP2,liao2019counting}.    Most relevant for this paper, Barvinok and Regts~\cite{barvinok2017weighted} give an FPTAS for $Z(G)$ for biregular, bipartite graphs with unequal degrees when the fugacity is sufficiently large, as an application of a much more general approximate counting result.

Here we give an FPTAS for the hard-core model on bipartite graphs of bounded degree whenever there is sufficient asymmetry in the degrees on either side or the fugacities assigned to the respective sides of the bipartition. In most biregular cases our results give significant improvement to the parameters from~\cite{barvinok2017weighted}, as we discuss below, and our algorithm does not require biregularity. 
More importantly, the method, based on the cluster expansion and the \kotecky-Preiss condition~\cite{kotecky1986cluster} and related to that of~\cite{helmuth2018contours, JenssenAlgorithmsSODA}, gives detailed probabilistic information about the hard-core model in addition to the algorithmic results.  In particular, we show that on bipartite graphs with parameters satisfying our conditions, the correlation between the occupancies of two vertices in the hard-core model decays exponentially fast in their distance.

The connection between convergence of the cluster expansion and decay of correlation for lattice spin models is well studied in statistical physics.  We hope that by illustrating the applicability of these techniques to general graphs and showing their connection to algorithmic results, we may encourage their adoption in computer science.

\subsection{Algorithmic results}

Let $\cG( \Delta_L, \del_R, \Delta_R)$ be the family of bipartite graphs $G$ with bipartition $(L,R)$ so that each $v \in L$ has degree  $d_v\le\Delta_L$ and each $v \in R$ has degree $d_v$ satisfying $\del _R \le d_v \le \Delta_R$.

We will consider the hard-core model on a graph $G \in \cG( \Delta_L, \del_R, \Delta_R)$ with fugacity $\lam_L$ assigned to each vertex $v \in L$ and fugacity $\lam_R$ assigned to each $v \in R$.  Our main algorithmic result is a sufficient condition on the imbalance of the graph (in terms of $\Delta_L, \del_R, \Delta_R, \lam_L, \lam_R$) to obtain an FPTAS for approximating $Z(G, \lam_L,\lam_R)$.  

\begin{theorem}
\label{thmAlgorithm}
Suppose
\begin{align}
\label{eqConditionMain}
6\Delta_L \Delta_R \lam_R &\le (1+\lam_L)^{\frac{\del_R}{\Delta_L}}   \, .
\end{align}
 Then there is an FPTAS for approximating the hard-core partition function $Z(G,\lam_L,\lam_R)$ and an efficient sampling scheme for sampling from $\mu_G$ for all $G \in \cG( \Delta_L, \del_R, \Delta_R)$ at fugacities $\lam_L$ and $\lam_R$. 
\end{theorem}

Some special cases of Theorem~\ref{thmAlgorithm} for the bivariate and univariate hard-core partition functions are given in the following corollary. Throughout, we assume all logarithms have base $e$. 
\begin{cor}
\label{corMain}
There is an FPTAS for approximating $Z(G)$ and an efficient sampling scheme for $\mu_G$ when:
\begin{enumerate}
\item $G$ is a $\Delta$-regular bipartite graph, and the fugacities $\lam_L$ and $\lam_R$ satisfy
\begin{align*}
\lam_L \ge 6 \Delta^2 \lam_R \, .
\end{align*}
\item $G$ is a biregular bipartite graph with degrees $\Delta_R> \Delta_L$, and the fugacity $\lam$ satisfies 
\begin{align*}
\lambda > (6\Delta_L \Delta_R)^\frac{\Delta_L}{\Delta_R - \Delta_L} \,.
\end{align*}
\item $G$ is a biregular bipartite graph with degrees $\Delta_L$ and $\Delta_R$ satisfying
\begin{align*}
\Delta_R \ge 7 \Delta_L \log (\Delta_L) \,,
\end{align*}
and the fugacity is $\lam=1$.
\end{enumerate}
\end{cor}

We can compare Corollary~\ref{corMain}, part (2),   to the condition given by Barvinok and Regts~\cite{barvinok2017weighted} for an FPTAS for biregular, bipartite graphs at fugacity $\lam$: $\Delta_R > \Delta_L$ and
\begin{align*}
\lam &\ge \left( 6.7 \sqrt{\Delta_R} \right)^{\Delta_L + \Delta_R}  \,.
\end{align*}
Our bound of $\lambda > (6\Delta_L \Delta_R)^\frac{\Delta_L}{\Delta_R - \Delta_L}$ gives a significant improvement except in the case when $\Delta_R =\Delta_L+1$. In fact, for $\Delta_L$ fixed, our bound on $\lam$ \textit{decreases} as $\Delta_R$ grows.

\subsection{Correlation decay}

We next show that for graphs satisfying condition~\eqref{eqConditionMain}, vertex-to-vertex correlations decay exponentially fast in their distance. We are able to do this because the convergence of the cluster expansion, which we use to obtain the FPTAS described above, also gives additional detailed probabilistic information about the model. 

Let $\mathbf I$ be a random independent set drawn from the hard-core model $\mu_G$.  For $v \in V(G)$, let $X_v = \mathbf 1_{v\in \mathbf I}$, and  $\mu_v = \E X_v= \Pr[ v\in \mathbf I]$.  For $u, v \in V(G)$ let $X_{uv} = \mathbf 1_{u,v \in \mathbf I}$, and $\mu_{uv} = \E X_{uv} = \Pr [u \in \mathbf I \wedge v \in \mathbf I]$. For $A \subseteq V(G)$, let $X_A = \mathbf 1_{A \subseteq \mathbf I}$, and $\mu_A = \Pr[A \subseteq \mathbf I]$.

\begin{theorem}
\label{thmCorDecay}
For $\Delta_L, \Delta_R, \del_R,  \lam_L, \lam_R$ satisfying~\eqref{eqConditionMain}, there exists constants $\eps >0, C>0$ so that the following holds.  

For all $G \in \cG( \Delta_L, \del_R, \Delta_R)$, and for all $u, v \in V(G)$, we have 
\begin{align}
\label{eqPointDecay}
\left | \mu_{uv} -\mu_u \mu_v  \right| &\le Ce^{- \eps D(u,v)} \,, 
\end{align}
where  $D(\cdot,\cdot)$ is the graph distance in $G$, and $\mu_G$ is the hard-core model on $G$ with fugacities $\lam_L, \lam_R$. 
  More generally, suppose $A \subset V(G), B \subset V(G)$. Then
\begin{align}
\label{eqSetDecay}
\left | \mu_{A \cup B} - \mu_A \mu_B  \right |&\le C' e^{-\eps D(A,B)} \,,
\end{align}
where the constant $C'$ depends on $\eps$ and $|A|, |B|, |N(A)|$, and $|N(B)|$. 
\end{theorem}

In fact, Theorem~\ref{thmCorDecay} follows from a bound on the \textit{truncated $m$-point correlation functions} (or \textit{semi-invariants}, or \textit{joint cumulants}), a measure of correlation decay that arises naturally in statistical mechanics (see e.g.~\cite{malyshev1980cluster,dobrushin1987completely}).

For $A= \{ v_1, \dots, v_k \} \subseteq V(G)$, the \textit{truncated correlation function} of $A$ is defined as
\begin{align*}
\kappa(A) &= \frac{\partial} {\partial t_{v_1}} \cdots \frac{\partial}{\partial t_{v_k}} \log \E \left [ e^{\sum_{v} t_v X_v }  \right]  \Bigg |_{t_v = 0, v\in V(G)} \, ,
\end{align*}
where the expectation is over the choice of random independent set from $\mu_G$. 

In particular, we can recover the vertex marginals and the correlation between two vertices, as one can show: 
\begin{align*}
\kappa(\{v\}) &= \mu_v \,,
\end{align*}
and 
\begin{align*}
\kappa ( \{u,v \}) &= \mu_{uv} - \mu_u \mu_v  \,.
\end{align*}

For a set of vertices $A \subseteq R$, let $\text{MST}(A)$ (minimum size Steiner tree) be the smallest $k$ so that $G$ contains a connected subgraph $H$  with $k$ edges containing all vertices of $A$.  
\begin{theorem}
\label{thmSemiInvariant}
If $\Delta_L, \del_R, \Delta_R, \lam_L, \lam_R$ satisfy~\eqref{eqConditionMain}, there exists $\eps>0$ so that for any $A \subseteq R$, 
\begin{align*}
|\kappa(A)|  &\le C e^{-\eps \text{MST}(A)} \,,
\end{align*}
where $C$ depends only on $\eps$ and $|A|$.  
\end{theorem}  
Theorem~\ref{thmSemiInvariant} immediately implies the bound~\eqref{eqPointDecay} in Theorem~\ref{thmCorDecay}  for $u,v \in R$ since $\text{MST}(\{u,v\}) = D(u,v)$.  We show how to derive the full conclusion of Theorem~\ref{thmCorDecay} from Theorem~\ref{thmSemiInvariant} in Section~\ref{secCorDecay}.

\subsection{Complex zeroes of the partition function}
\label{secZeroes}
The polynomial interpolation method, Barvinok's approach to approximate counting~\cite{barvinok2017combinatorics,barvinok2017weighted}, involves deducing convergence of the Taylor series for the logarithm of a univariate partition function from the existence of a zero-free region of the partition function in the complex plane.  This convergence, along with the algorithm of Patel and Regts~\cite{patel2016deterministic} for efficient computation of low-order coefficients of a partition function, leads to an FPTAS in a wide variety of approximate counting problems. 

Our approach of truncating the cluster expansion was inspired by and related to this approach.  In fact, convergence of the cluster expansion implies that the partition function does not vanish, and while our algorithmic approach does not require it (and the preceding algorithmic and probabilistic theorems are stated for positive fugacities), the \kotecky-Preiss convergence criteria works naturally with complex fugacities.  Thus we can deduce the following result on zeroes of the bivariate hard-core partition function.
\begin{theorem}
Suppose $\Lam_L, \Lam_R >0$ and 
\begin{align}
\label{eqConditionMainComplex}
6\Delta_L \Delta_R \Lam_R &\le (1+\Lam_L)^{\frac{\del_R}{\Delta_L}}   \, .
\end{align}
Then for all $G \in \cG ( \Delta_L, \del_R, \Delta_R )$ and all  $\lam_L, \lam_R \in \mathbb C$ satisfying 
\begin{align*}
| \lam_R | \le \Lam_R \\
| 1+ \lam_L| \ge 1+\Lam_L,
\end{align*}
the bivariate hard-core partition function satisfies
$$Z(G,\lam_L,\lam_R) \ne 0 \,. $$
\label{thmZero}
\end{theorem}

\noindent Theorem~\ref{thmZero} follows from the proof of Theorem~\ref{thmAlgorithm} with $\lambda_L, \lambda_R$ replaced with $\Lambda_L, \Lambda_R$; we give the details in Section~\ref{secCluster}.

\subsection{Discussion}
\label{secDiscuss}

The cluster expansion is a {\it perturbative} technique, based on expressing a partition function in terms of deviations from a simple, easy to understand \textit{ground state}.  In high temperature (small $\lam$) regimes, this ground state is the empty independent set.  Barvinok's polynomial interpolation method is also a perturbative technique, and interpolating from $\lam=0$ is akin to measuring deviations from the empty independent set.  

On the other hand, previous algorithmic applications of the cluster expansion in low temperature (large $\lam$) regimes~\cite{helmuth2018contours,JenssenAlgorithmsSODA} considered systems with multiple ground states, e.g. the all $L$ or all $R$ occupied independent sets for the hard-core model on a bipartite graph.  A necessary first step in these cases is to show that the partition function of the entire system is well approximated by the sum of partition functions representing deviations from each ground state. 
 The asymmetric setting of this paper is in fact a simpler low-temperature case than those previously considered cases.  Condition~\eqref{eqConditionMain} ensures sufficient asymmetry in a bipartite graph that the hard-core model can be expressed in terms of deviations from a single ground state; this ground state is dominant enough that $Z(G)$ is well approximated by small deviations from it. This simplifies the general low-temperature argument, as the first approximation step showing $Z(G)$ is well-approximated by the sum of multiple partition functions is not necessary.
 
 It would be interesting to see if other known methods for approximate counting -- Markov chain Monte Carlo or the correlation decay method -- can be used to obtain an FPTAS for a similar range of parameters as we do here.

\section{Convergence of the cluster expansion}
\label{secCluster}

Following~\cite{helmuth2018contours,JenssenAlgorithmsSODA}, our algorithms will be based on approximating the partition function of a \textit{polymer model}~\cite{gruber1971general,kotecky1986cluster} using the cluster expansion (for a textbook introduction to both polymer models and the cluster expansion see Chapter 5 of~\cite{friedli2017statistical}).

A polymer model consists of a set $\cP$ of abstract objects we call \textit{polymers}.  Each polymer $\gamma$ is equipped with a real or complex-valued weight $w_\gamma$, and there is a symmetric compatibility relation on $\cP$; we write $\gamma \sim \gamma'$ if $\gamma$ and $\gamma'$ are compatible, and $\gamma \nsim \gamma'$ if they are incompatible.  We require that $\gamma \nsim \gamma$ for all $\gamma \in \cP$.  

The partition function of a polymer model is
\begin{align*}
\Xi(\cP)&= \sum_{\substack{\Gamma \subseteq \cP \\ \text{compatible}}} \prod_{\gamma \in \Gamma} w_{\gamma} 
\end{align*}
where the sum is over all pairwise compatible collections of polymers (the empty collection contributes $1$ to the sum).   If the weights $w_\gamma$ are real and non-negative then we can define a probability measure $\nu$ on $\Omega$, with
\begin{align}
\label{eqNu}
\nu(\Gamma) &= \frac{  \prod_{\gamma \in \Gamma} w_{\gamma}  }{ \Xi(\cP)}  \, .
\end{align}

The cluster expansion is an infinite series representation of $\log  \Xi (\cP)$. A  \textit{cluster} $\Gamma$ is an ordered multiset of polymers from $\cP$ whose \textit{incompatibility graph} $H(\Gamma)$ (a vertex for every polymer and an edge between each pair of incompatible polymers) is connected.  Denote by $\cC$ the set of all clusters from $\cP$.   As a formal power series in the polymer weights, the cluster expansion is
\begin{align*}
\log \Xi (\cP) &= \sum_{\Gamma \in \cC} w(\Gamma) \, ,
\end{align*}
where
\begin{align*}
w(\Gamma) &= \phi(H(\Gamma)) \prod_{\gamma \in \Gamma} w_{\gamma} \, ,
\end{align*}
and the function $\phi(H)$ is the \textit{Ursell function}
\begin{align*}
\phi(H) &= \frac{1}{|V(H)|!} \sum_{\substack{ A \subseteq E(H) \\ \text{spanning, connected} }} (-1)^{|A|} \,.
\end{align*}

A sufficient condition for the convergence of the cluster expansion is given by a theorem of \kotecky{ }and Preiss~\cite{kotecky1986cluster}.    
\begin{theorem}
\label{thmKP}
Let $a : \cP \to [0,\infty)$ and $b: \cP \to [0,\infty)$ be two given functions and suppose that for all $\gamma \in \cP$, 
\begin{align}
\label{eqKPcond1}
\sum_{\gamma' \nsim \gamma}  |w_{\gamma'}| e^{a(\gamma') + b(\gamma')}  &\le a (\gamma) \,,
\end{align}
then the cluster expansion converges absolutely, and, moreover, for all $\gamma \in \cC$ we have
\begin{equation}
\label{eqKPtail1}
\sum_{\substack{\Gamma \in \cC :\\ \exists \gamma' \in \Gamma, \gamma' \nsim \gamma }} \left |  w(\Gamma) \right| e^{b( \Gamma)} \le a(\gamma) \,,
\end{equation}
where 
\begin{align*}
b(\Gamma) = \sum_{\gamma \in \Gamma} b(\gamma) \,.
\end{align*}
\end{theorem}

Now we define a polymer model representation of the hard-core model on a bipartite graph and use the \kotecky-Preiss condition to obtain our algorithmic and correlation decay results.
We fix bipartite graph $G$ with bipartition $(L,R)$ and respective minimum and maximum degrees $\del_L, \Delta_L$ and $\delta_R, \Delta_R$.  Fix also activities $\lam_L$ and $\lam_R$ for vertices in $L$ and $R$ respectively. 

We define a polymer $\gamma$ to be a $2$-linked subset of $R$; that is, $\gamma \subseteq R$ is connected in the graph $G^2$ in which each vertex in $V(G)$ is joined to all vertices within distance $2$ of it in $G$.  Let $\cP=\cP(G)$ be the set of all polymers of $G$.  Two polymers $\gamma, \gamma'$ are compatible if $\gamma \cup \gamma'$ is not $2$-linked and incompatible otherwise.  In particular, as required, $\gamma \nsim \gamma$ for every polymer $\gamma$.   For each polymer $\gamma$ we define its weight function to be 
\begin{align*}
w_\gamma = \frac{\lam_R^{|\gamma|}}{(1+\lam_L)^{|N(\gamma)|}} \,.
\end{align*}
Let $\Xi (\cP)$ be the corresponding polymer model partition function.  This polymer model is almost the same as the polymer models used in~\cite{JKP2,liao2019counting} to approximate the hard-core partition function on random regular bipartite graphs; the difference being that in those cases, two polymer models were defined, one representing $L$-dominant independent sets and one representing $R$-dominant independent sets, and the sum of their partition functions was shown to be a good approximation of $Z(G)$.  Here in our asymmetric setting, one polymer model suffices, and in fact the polymer model partition function $\Xi(\cP)$ is, up to scaling, {exactly} the hard-core partition function.  

\begin{lemma}
\label{lemPPart}
\begin{align}
\label{eqPartitionFunc}
Z(G) &= (1+\lam_L)^{|L|} \Xi(\cP) \, .
\end{align}
\end{lemma}
\begin{proof}
   To see this, consider an independent set $I \in \cI (G)$, and decompose $I \cap R$ into its maximal $2$-linked components, call these $\gamma_1, \dots, \gamma_k$.  Each of these is a polymer since it is $2$-linked, and they are pairwise compatible since each is maximal.   Moreover, the number of vertices in $L$ that are blocked from being in an independent set by a vertex in $I \cap R$ is $\sum_{i=1}^k |N(\gamma_i)|$, and so if we sum over all independents sets $I'$ so that $I'\cap R = I \cap R$, the contribution to $Z(G)$ is $(1+\lam_L)^{|L|} \prod_{i=1}^k \frac{ \lam_R^{|\gamma_i|}}{(1+\lam_L)^{|N(\gamma_i)|}}$. Summing over all possible sets $I \cap R$, we obtain~\eqref{eqPartitionFunc}.  
\end{proof}

Moreover, as the proof of Lemma~\ref{lemPPart} indicates, if we let $\nu$ be the probability measure on sets of compatible polymers defined by~\eqref{eqNu}, then it is easy to recover a sample $\mathbf I$  from $\mu_G$ given a sample $\mathbf \Gamma$ from $\nu$: include all vertices of the polymers of $\mathbf \Gamma$ in $\mathbf I$ and in addition, for each vertex $v \in L$ that is not blocked from being in $\mathbf I$ by one of these vertices, include it in $\mathbf I$ independently with probability $\frac{\lam}{1+\lam}$.

Next we turn to using the cluster expansion algorithmically. 
We define the size of a cluster $\Gamma$ of polymers as the sum of their sizes:
\begin{align*}
|\Gamma| = \sum_{\gamma \in \Gamma} |\gamma| \,.
\end{align*}
With this definition, we can define a truncated cluster expansion that only sums over clusters up to a certain size which we will use for the approximate counting algorithm.
\begin{align*}
T_m(\cP) \coloneqq  \sum_{\substack{\Gamma \in \cC \\ |\Gamma| < m}} w(\Gamma) \,.
\end{align*}
We can then ask when $T_m(\cP)$ is a good additive approximation of $\log \Xi(\cP)$.

\begin{lemma}
\label{lemKPapplied}
Suppose that for some $\eta>0$ and every $v \in R$,
\begin{align}
\label{eqKP2}
\sum_{\gamma \ni v} |w_{\gamma}| e^{(1/2+\eta) |\gamma|} &\le \frac{1}{2 (\Delta_R (\Delta_L-1) +1)} \, ,
\end{align}
then the cluster expansion converges absolutely, the cluster weights satisfy
\begin{align}
\label{eqClusterTailBound}
\sum_{\substack{ \Gamma \in \cC \\ \Gamma \ni v}} | w(\Gamma)| e^{\eta |\Gamma|} &\le 1 \,,
\end{align}
and, in particular,
\begin{align}
\label{eqClusterTruncateBound}
\left|  T_m (\cP) - \log  \Xi (\cP) \right |  &\le   |R| e^{-m \eta } \,.
\end{align}
\end{lemma} 
\begin{proof}
  We will apply Theorem~\ref{thmKP} with the functions $a(\gamma) = |\gamma|/2$ and $b(\gamma) = \eta |\gamma|$.  
Because any polymer incompatible with $\gamma$ must include a vertex of $\gamma$ or a vertex at distance 2 from $\gamma$, for our polymer model we see that
\begin{align*}
\sum_{\gamma' \not\sim \gamma} |w_{\gamma'}| e^{ a(\gamma') + b(\gamma')}  \leq \sum_{\substack{v \in R\\ D(v,\gamma) \leq 2}} \sum_{\gamma' \ni v} | w_{\gamma'} | e^{ (1/2 + \eta) |\gamma'|}.
\end{align*}
 The number of vertices in $R$ with $D(v,\gamma) \leq 2$ is at most $(\Delta_R(\Delta_L-1) + 1) |\gamma|$, and so
 \begin{align*}
\sum_{\gamma' \not\sim \gamma} |w_{\gamma'}| e^{ a(\gamma') + b(\gamma')}  &\leq (\Delta_R(\Delta_L-1) + 1) |\gamma| \max_{v \in R} \sum_{\gamma' \ni v} | w_{\gamma'} | e^{ (1/2 + \eta) |\gamma'|} \\
&\le | \gamma|/2  = a(\gamma) \,.
\end{align*}
by~\eqref{eqKP2}, giving~\eqref{eqKPcond1}. Thus by Theorem~\ref{thmKP} the cluster expansion for $\log \Xi(\cP)$ converges absolutely.

	To obtain the bound~\eqref{eqClusterTailBound}, we augment the set of polymer we consider and use Theorem~\ref{thmKP} again. We adjoin to $\cP$ a polymer $\gamma_v$ for each $v \in R$, with $w_{\gamma_v} =0$ and $\gamma_v \nsim \gamma$ for each $\gamma \ni v$. We set $a(\gamma_v) =1$ and  $b(\gamma_v) =0$. To verify the \kotecky{}-Preiss condition~\eqref{eqKPcond1} still holds for $\cP$ with these new polymers added, we must also verify that for all $v \in R$,	
	\begin{align*}
\sum_{\gamma \ni v}  |w_{\gamma}| e^{(1/2+\eta) |\gamma|} &\le 1 \, .\end{align*}
This is implied by~\eqref{eqKP2}.  
Using conclusion~\eqref{eqKPtail1} of Theorem~\ref{thmKP} for the adjoined polymers, we see that for all $v \in R$, 
\begin{align*}
\sum_{\substack{ \Gamma \in \cC \\ \Gamma \ni v}} |w(\Gamma)| e^{\eta |\Gamma|} 
\le \sum_{\substack{\Gamma \in \cC :\\ \exists \gamma' \in \Gamma, \gamma' \nsim \gamma_v }} \left |  w(\Gamma) \right| e^{b( \Gamma)} \le a(\gamma_v) = 1 \,.  
\end{align*}

	The bound~\eqref{eqClusterTruncateBound} follows from~\eqref{eqClusterTailBound} by  observing that
	\begin{align}
	\label{eqGamvBound}
\sum_{\substack{ \Gamma \in \cC \\ \Gamma \ni v \\ |\Gamma| \ge m}} | w(\Gamma)|  
&\le  e^{-\eta m} \sum_{\substack{ \Gamma \in \cC \\ \Gamma \ni v \\ |\Gamma| \ge m}} | w(\Gamma)| e^{\eta |\Gamma|}
\le e^{-\eta m} \,,
\end{align}
	and then summing over all $v \in R$. 
	\end{proof}

With this sufficient condition for the convergence of the cluster expansion for $\log \Xi (\cP)$ we can prove Theorem~\ref{thmAlgorithm}. 

\begin{proof}[Proof of Theorem~\ref{thmAlgorithm}]
	Let $G \in \cG(\Delta_L, \delta_R, \Delta_R)$, and let $n = |R|$. Let $\lambda_L$ and $\lambda_R$ be such that~\eqref{eqConditionMain} is satisfied.   Let $\cP = \cP(G)$ be the set of polymers defined above (without the additional polymers $\gamma_v$ that were adjoined in the previous proof).
	
	We first consider the approximate counting algorithm.  By Lemma~\ref{lemPPart}, to approximate $Z(G)$ it suffices to obtain an $\eps$-relative approximation to $\Xi(\cP)$; we can then get an $\eps$-relative approximation to $Z(G)$ by multiplying by $(1+\lambda_L)^{|L|}$. 

The bound~\eqref{eqClusterTruncateBound} suggests an algorithm for approximating $\Xi(\cP)$ when~\eqref{eqKP2} holds by calculating $T_m(\cP)$ for $m = \log(n/\eps)/\eta$ and exponentiating. This can be done as follows:
\begin{enumerate}
	\item Enumerate all clusters $\Gamma \in \cC$ with $|\Gamma| < m$; call the list of such clusters $\cC_m$.
	\item For each cluster $\Gamma \in \cC_m$, compute $\phi(H(\Gamma))$ and $\prod_{\gamma \in \Gamma} w_\gamma$. 
	\item Compute $T_m(\cP)$ by summing: 
	\[ T_m(\cP) = \sum_{\substack{\Gamma \in \cC_m}} \phi(H(\Gamma)) \prod_{\gamma \in \Gamma} w_\gamma\,.\] 
	\item Output $\exp(T_m(\cP))$. 
\end{enumerate}
It is shown in~\cite{JKP2}, using ideas and tools from~\cite{helmuth2018contours,patel2016deterministic,bjorklund2008computing}, that this algorithm can be implemented with running time $O \left(n \cdot (n/\eps)^{O (\log (\Delta_L \Delta_R)/\eta    )}   \right)$, which for $\Delta_L, \Delta_R$ fixed is polynomial in $n$ and $1/\eps$. 

What remains is to show that~\eqref{eqKP2} holds for some $\eta>0$.  
By double counting the edge boundary of a polymer $\gamma$, we have that $|N(\gamma)| \geq \frac{\delta_R}{\Delta_L} |\gamma|$, and so 
\begin{align*}
|w_\gamma| &\le \left( \frac{\lam_R}{(1+\lam_L)^{\frac{\delta_R}{\Delta_L}}} \right) ^{|\gamma|} \,.
\end{align*}

Then using the fact from~\cite{BorgsChayesKahnLovasz} that the number of $2$-linked subsets of $R$ containing a fixed vertex $v$ of size $k$ is at most $ \frac{ (e\Delta_R(\Delta_L-1))^{k-1}}{k^{3/2}}$, we have 
\begin{align*}
\sum_{\gamma \ni v}  |w_\gamma| e^{(1/2+ \eta)|\gamma|}
&\leq \sum_{k = 1}^{\infty} \frac{ (e\Delta_R(\Delta_L-1))^{k-1}}{k^{3/2}} \left( \frac{\lam_R}{(1+\lam_L)^{\frac{\delta_R}{\Delta_L}}} \right) ^{k} e^{(1/2 + \eta)k}\\
&\le \sum_{k = 1}^{\infty} \frac{ (e(\Delta_R(\Delta_L-1)+1))^{k-1}}{k^{3/2}} \left( \frac{\lam_R}{(1+\lam_L)^{\frac{\delta_R}{\Delta_L}}} \right) ^{k} e^{(1/2 + \eta)k}\\
&= \frac{1}{e (\Delta_R (\Delta_L-1)+1)} \sum_{k = 1}^{\infty}  \frac{1}{k^{3/2}}\left( \frac{(\Delta_R (\Delta_L-1)+1) \lambda_R \, e^{3/2 +\eta}}{(1+\lambda_L)^{\frac{\delta_R}{\Delta_L}} } \right)^k.
\end{align*}
Thus it suffices to show that 
\begin{align*}
 \sum_{k = 1}^{\infty}  \frac{1}{k^{3/2}}\left( \frac{(\Delta_R (\Delta_L-1) +1)\lambda_R e^{3/2 +\eta}}{(1+\lambda_L)^{\frac{\delta_R}{\Delta_L}} } \right)^k &\le \frac{e}{2} \,.
\end{align*}
Since $\sum_{k\ge 1} s^k/k^{3/2} < e/2$ for $0\le s\le.832$, 
it is enough to show  that 
 \begin{align*}
 \frac{(\Delta_R (\Delta_L-1) +1)\lambda_R e^{3/2 +\eta}}{(1+\lambda_L)^{\frac{\delta_R}{\Delta_L}} }\leq .832 \,,
\end{align*}
or 
 \begin{align*}
 \frac{(\Delta_R (\Delta_L-1) +1)\lambda_R e^{\eta}}{(1+\lambda_L)^{\frac{\delta_R}{\Delta_L}} }\leq \frac{.832}{e^{3/2}} \approx .1856  \,\,.
\end{align*}
In particular, if~\eqref{eqConditionMain} holds then
\begin{align*}
 \frac{(\Delta_R (\Delta_L-1) +1)\lambda_R}{(1+\lambda_L)^{\frac{\delta_R}{\Delta_L}} }\leq \frac{1}{6} \,,
\end{align*}
and so we can obtain~\eqref{eqKP2} with $\eta =.1 < \log (6 \cdot .1856) $.  Applying the algorithm of~\cite{JKP2} gives the FPTAS for $Z(G)$.   

Next we turn to approximate sampling.  The usual approach to using an approximate counting algorithm to obtain an approximate sampling algorithm is via self-reducibility~\cite{jerrum1986random}.  We cannot directly apply this in our setting, however, as reducing a graph $G \in \cG(\Delta_L, \delta_R, \Delta_R)$ might result in a graph $G' \notin \cG(\Delta_L, \delta_R, \Delta_R)$ because of the minimum degree condition.  Instead we can apply self-reducibility on the level of polymers, using~\cite[Theorem 5.1]{helmuth2018contours}, which provides a general reduction of approximate sampling to approximate counting for polymer models.  The high-level idea of the algorithm, applied to our setting, is to create a compatible polymer configuration $\Gamma$ one polymer at a time, with a final distribution close to $\nu$ given in~\eqref{eqNu}. Once we have such a configuration $\Gamma$ we can extend this to an independent set $I \in \cI(G)$ as described in the remark after Lemma~\ref{lemPPart}.

To approximately sample from $\nu$, we order the vertices in $R$ arbitrarily, and then one-by-one we determine which, if any, polymer containing a given vertex $v$ is present in $\Gamma$.  Based on the previous choices, some polymers in $\cP$ are excluded from future choices as their addition would form an incompatible pair of polymers, and so the calculation of marginal probabilities involves approximating $\Xi(\cP')$, the partition function of a set of polymers $\cP' \subset \cP$.  While the polymer model associated to an arbitrary set $\cP' \subset \cP$ may not map to the hard-core model on any graph, this polymer model automatically satisfies~\eqref{eqKPcond1} if the original polymer model does, since the condition only becomes weaker on removing polymers.  Thus we can efficiently approximate $\Xi(\cP')$ using the algorithm from~\cite{JKP2}, and thus sample efficiently from $\nu$. 
\end{proof}

Next we prove Corollary~\ref{corMain}.

\begin{proof}[Proof of Corollary~\ref{corMain}]
Parts (1) and (2) follow from simple substitutions into~\eqref{eqConditionMain}. 

For (3), we set $\Delta_L= \Delta$ and $\Delta_R = c \Delta \log (\Delta)$, and ask for which $c$ we have
\begin{align}
\label{eqcor3bound}
 c \log2 \log \Delta -2 \log \Delta - \log \log \Delta - \log c - \log 6   \ge 0  \,,
\end{align}
which implies~\eqref{eqConditionMain} is satisfied. 
 Some simple calculus shows that~\eqref{eqcor3bound} is satisfied for $\Delta \ge 2$ when $c \ge 11$.  But if we use the fact that for $\Delta \le 5$, an FPTAS is given by Liu and Lu~\cite{liu2015fptas}, then we see that~\eqref{eqcor3bound} is satisfied for $\Delta \ge 6$ when $c \ge 7$.  
\end{proof}

Finally we prove Theorem~\ref{thmZero}, which follows almost immediately from the proof of Theorem~\ref{thmAlgorithm}.

\begin{proof}[Proof of Theorem~\ref{thmZero}]
Condition~\eqref{eqConditionMainComplex} is the same as Condition~\eqref{eqConditionMain} with $\lambda_L$ and $\lambda_R$ replaced with $\Lambda_L$ and $\Lambda_R$.  If $\lambda_L$ and $\lambda_R$ are complex such that $|\lambda_R| \leq \Lambda_R$ and $|1+\lambda_L| \geq 1 + \Lambda_L$, then for any polymer $\gamma$,
\begin{align*}
|w_\gamma|
=  \left|\frac{\lam_R^{|\gamma|} }{(1+\lam_L)^{|N(\gamma)|}}\right|
 &\le \left( \frac{\Lam_R}{(1+\Lam_L)^{\frac{\delta_R}{\Delta_L}}} \right) ^{|\gamma|} \,.
\end{align*}
Condition~\eqref{eqConditionMainComplex} for $\Lambda_L$ and $\Lambda_R$ then suffices for the analysis in the proof of Theorem~\ref{thmAlgorithm} to show Condition~\eqref{eqKP2} holds for $\lambda_L$ and $\lambda_R$. 
Lemma~\ref{lemKPapplied} then implies the cluster expansion for $\lambda_L$ and $\lambda_R$ converges absolutely. 
 Because the cluster expansion is a convergent power series for $\log \Xi(\cP)$, for these $\lambda_L$ and $\lambda_R$ the polymer partition function $\Xi(\cP)$ cannot be zero and so $Z(G) = (1 + \lambda_L)^{|L|} \Xi(\cP)$ is also nonzero. 
\end{proof}

\section{Correlation decay}
\label{secCorDecay}

As Dobrushin showed~\cite{dobrushin1996estimates,dobrushin1996perturbation} the cluster expansion and the \kotecky-Preiss convergence condition are very well suited to studying the cumulants and joint cumulants of random variables that are functions of polymer configurations.  

It follows almost directly from estimates such as~\eqref{eqClusterTailBound} that the magnitudes of  joint cumulants of random variables that depend on disjoint polymers decay exponentially in the distance between these polymers (Theorem~\ref{thmSemiInvariant}). The joint cumulant of a set of random variables  vanishes if there is a non-trivial partition of the set into two sets of random variables that are  independent of each other, and this holds approximately as well. Showing joint cumulants decay exponentially in the distance between polymers shows a form of decay of correlations, and in fact implies the perhaps more familiar form of~\eqref{eqSetDecay}.

While the techniques and calculations described here are very similar to those in~\cite{dobrushin1996estimates}, and the properties of joint cumulants are standard facts in the study of statistical mechanics on lattices, we present the proofs in a self-contained way in order to emphasize their elementary nature and to encourage the use of these methods in computer science and in non-lattice settings.

We again fix a graph $G$ and consider the associated polymer model with polymers $\cP$ and partition function $\Xi$.  We suppose that~\eqref{eqKP2} holds for some $\eta>0$ and calculate exact expressions for the joint cumulants in our setting using the cluster expansion.  We begin by introducing auxiliary polymer weights as follows.  Given variables $t_v$, $v\in R$, let
\begin{align*}
\tilde w_{\gamma} &= w_\gamma e^{ \sum_{v \in \gamma } t_v} \,, 
\end{align*}
and let $\tilde \Xi$ be the polymer model partition function derived from these weights.  Let $X_v= \mathbf 1_{v\in \mathbf I}$ for $\mathbf I $ drawn from $\mu_G$ and recall that the distribution of $\mathbf I \cap R$ (the set of occupied vertices in $R$) is identical to the distribution of $\bigcup_{\gamma \in \mathbf \Gamma} \gamma$ (the union of polymers) for a polymer configuration $\mathbf \Gamma$ drawn from $\nu$.  Using this we can write
\begin{align*}
\E e^{\sum_{v \in R} t_v X_v} &=  \sum_{\substack{\Gamma \subseteq \cP \\ \text{compatible}}} \nu(\Gamma) \prod_{\gamma \in \Gamma} e^{ \sum_{v \in \gamma } t_v} \\
&=  \sum_{\substack{\Gamma \subseteq \cP \\ \text{compatible}}} \frac{1}{\Xi} \prod_{\gamma \in \Gamma} w_\gamma e^{ \sum_{v \in \gamma } t_v} \\
 &= \frac{ \tilde \Xi}{ \Xi} \,.
\end{align*}

For a set $A \subseteq R$  we can then write the joint cumulant of $A$ as
\begin{align*}
\kappa(A) &:= \frac{\partial^{|A|}  \log \E e^{\sum_{v \in R} t_v X_v}  }{ \prod_{u \in A} \partial t_u  }   \Bigg |_{t_v = 0, v\in V(G)} \\ 
&=  \frac{\partial^{|A|}   \log \tilde \Xi   }{ \prod_{u \in A} \partial t_u  }   \Bigg |_{t_v = 0, v\in V(G)} \, .
\end{align*}
For a vertex $v \in R$, and a cluster $\Gamma$, define $Y_v(\Gamma)$ to be the number of polymers in $\Gamma$ containing $v$: 
\begin{align*}
Y_v(\Gamma) &= \sum_{\gamma \in \Gamma} \mathbf 1_{v \in \gamma} \,.
\end{align*}
Then using the cluster expansion for $\log \tilde \Xi$ (and the fact that it converges absolutely if $t_v \le \eta$ for all $v$), we have
\begin{align}
\nonumber
\kappa(A) &=  \frac{\partial^{|A|}     }{ \prod_{u \in A} \partial t_u  }  \sum_{\Gamma \in \cC} w(\Gamma)  \Bigg |_{t_v = 0, v\in V(G)} \\
\label{eqClusterKappa}
 &= \sum_{\Gamma \in \cC} w(\Gamma) \prod_{v\in A} Y_v(\Gamma) \,. 
\end{align}
From this expression we can immediately derive Theorem~\ref{thmSemiInvariant}.
\begin{proof}[Proof of Theorem~\ref{thmSemiInvariant}]
  We will show that for $A \subseteq R$,
\begin{align}
\label{eqClusterAbound}
\sum_{\Gamma \in \cC} |w(\Gamma) | \prod_{v \in A } Y_v(\Gamma) &\le  C  e^{-\eta  \text{MST}(A)/2}
\end{align}
where $C$ depends on $\eta$ and $|A|$.   The size of a cluster $\Gamma$ containing all vertices of $A$ is at least $\text{MST}(A)/2$, where $\text{MST}(A)$ is the size of the minimum Steiner tree in $G$ containing $A$; this Steiner tree will contain vertices in both $L$ and $R$, but at least half of its vertices will be in $R$. More generally suppose a cluster $\Gamma$ contains $Y_v(\Gamma) \ge 1$ copies of $v$ for each $v \in A$.  Then 
\begin{align}
\label{eqGammaLB}
|\Gamma | \ge \text{MST}(A)/2 + \sum_{v \in A} (Y_v(\Gamma) -1) 
\end{align}
since each additional copy of a vertex contributes $1$ to the size of $\Gamma$.

We arrange the sum on the LHS of~\eqref{eqClusterAbound} based on the number of copies of each $v \in A$:
\begin{align*}
\sum_{\substack{\Gamma \in \cC }} |w(\Gamma)| \prod_{v \in A} Y_v(\Gamma) &= \sum_{ \{y_v \ge 1 \}_{v\in A} } \sum_{\substack{\Gamma \in \cC  \\ Y_v(\Gamma) = y_v \forall v \in A}} |w(\Gamma)| \prod_{v \in A} y_v \,. 
\end{align*}
Now since Condition~\eqref{eqKP2} holds with constant $\eta>0$, we can use \eqref{eqGamvBound} and~\eqref{eqGammaLB}:
\begin{align*}
\sum_{\substack{\Gamma \in \cC }} |w(\Gamma)| \prod_{v \in A} Y_v(\Gamma)  &\le \sum_{ \{y_v \ge 1 \}_{v\in A} } e^{-\eta ( \text{MST}(A)/2 + \sum_{v \in A} (y_v -1))} \prod_{v \in A} y_v \\
&\le e^{-\eta  \text{MST}(A)/2} \sum_{ \{y_v \ge 1 \}_{v\in A} } e^{-\eta  \sum_{v \in A} (y_v -1))} \prod_{v \in A} y_v \\
&\le C  e^{-\eta  \text{MST}(A)/2} \,,
\end{align*}
and so the conclusion of Theorem~\ref{thmSemiInvariant} holds with $\eps = \eta/2$ and  
$$C =\sum_{ \{y_v \ge 1 \}_{v\in A} } e^{-\eta  \sum_{v \in A} (y_v -1))} \prod_{v \in A} y_v < \infty \, .$$ 
\end{proof}

To prove Theorem~\ref{thmCorDecay} we need to do two things.  First we need to show that the exponential decay of joint cumulants in the minimum Steiner tree size implies that $|\mu_{A \cup B} - \mu_A \mu_B|$ decays exponentially in the distance from $A$ to $B$; second we need to extend these results from sets of vertices contained solely in $R$ to general sets of vertices in $L \cup R$.

\begin{lemma}
\label{lemBasisChange}
Suppose $A \subset R$ and $B \subset R$ are disjoint.  Suppose also that for any $S \subseteq A \cup B$ so that $S \cap A \ne \emptyset$ and $S \cap B \ne \emptyset$, we have 
\begin{align*}
|\kappa(S)|  &\le \eps.  
\end{align*}
Then 
\begin{align*}
\left |\mu_{A \cup B} -\mu_A \mu_B    \right | &\le C \eps
\end{align*}
where  $C$ depends only on $|A|$ and $|B|$. 
\end{lemma}
\begin{proof}  
We will use the following formula (see e.g.~\cite{leonov1959method}):
\begin{align*}
\mu_A &= \sum _{\pi \in \Sigma(A)} \prod_{S \in \pi} \kappa(S)  \,,
\end{align*}
where $\Sigma (A)$ is the family of all set partitions of $A$.  
This gives
\begin{align*}
\mu_{A \cup B} - \mu_A \mu_B &=  \sum _{\pi \in \Sigma(A \cup B)} \prod_{S \in \pi} \kappa(S) - \left( \sum _{\pi \in \Sigma(A)} \prod_{S \in \pi} \kappa(S) \right) \left (\sum _{\pi \in \Sigma(B)} \prod_{S \in \pi} \kappa(S) \right) \,.
\end{align*}
Now if we restrict the first sum on the RHS to those partitions whose parts are either entirely within $A$ or entirely within $B$, then we obtain exactly  $\mu_A \mu_B$ which cancels the second term on the RHS, and so we obtain (for $|\pi|$ the number of sets in partition $\pi$):
\begin{align*}
\mu_{A \cup B} - \mu_A \mu_B &= \sum _{\substack{\pi \in \Sigma(A \cup B)\\ \exists S' \in \pi: S' \cap A \ne \emptyset, S' \cap B \ne \emptyset}} \prod_{S \in \pi} \kappa(S)\,. 
\end{align*}
Now since $|\kappa(S')| \le \eps$ for $S'$ that intersects both $A$ and $B$, and the partitions in the above sum have at least one such $S'$,  we have 
\begin{align*}
\left |  \mu_{A \cup B} - \mu_A \mu_B \right | &\le C \eps 
\end{align*}
where $C$ is the number of set partitions of $A \cup B$,  times the maximum size of a joint cumulant of at most $|A| + |B|$  indicator random variables raised to the $|A| + |B|$ power, which depends only on $|A|$ and $|B|$~\cite{leonov1959method}.
\end{proof}

With these ingredients we prove Theorem~\ref{thmCorDecay}. 

\begin{proof}[Proof of Theorem~\ref{thmCorDecay}]
It suffices to prove the second statement in the theorem as the first follows by taking $A= \{u \}, B =\{v\}$.    First we consider the special case $A \subset R, B \subset R$.  

For any $S \subseteq R$ that intersects both $A$ and $B$, we have $\text{MST}(S) \ge D(A,B)$.  Then from Theorem~\ref{thmSemiInvariant} we have $|\kappa(S)| \le C e^{-\eta D(A,B)/2}$ for such $S$.  Then applying Lemma~\ref{lemBasisChange}, we obtain 
\begin{align*}
\left | \mu_{A \cup B}  - \mu_A \mu_B \right | &\le C' e^{-\eta D(A,B)/2} \,.
\end{align*}

Now we consider the general case $A \subset V(G), B \subset V(G)$.    We assume that $D(A, B)>2$ as otherwise the statement is trivial.  We can also assume that both $A$ and $B$ are independent sets in $G$, as otherwise both $\mu_{A \cup B}$ and $\mu_A \mu _B$ are $0$.    We say a vertex $v$ is \textit{unblocked} by the independent set $I$ if $N(v) \cap I = \emptyset$.  Let $\hat \mu_v = \Pr[ v \text{ unblocked}]$, and for a set of vertices $S$, $\hat \mu_S = \Pr[ S \text{ unblocked}]$.  Then if $S$ is an independent set itself, 
\begin{align*}
\mu_S = \left( \frac{\lam}{1+\lam} \right) ^{|S|} \hat \mu _S \,. 
\end{align*}
To see this, note that $S$ must be unblocked to be in the independent set.  Given that $S$ is unblocked, since $S$ is independent itself, the probability that each vertex $v\in S$ is in the independent set is $\frac{\lam}{1+\lam}$ and these events are conditionally independent over the  vertices in $S$.  

For $S \subset V(G)$, we introduce the notation $\rho_S = \Pr[ S \cap \mathbf I = \emptyset]$. If $S \subset L$, by definition, $\hat \mu_S = \rho_{N(S)}$.  
Using inclusion exclusion, we have, for $T \subset R$,
\begin{align*}
\rho_T &=   \sum_{Q \subseteq T} (-1)^{|Q|} \mu_Q \,,
\end{align*}
where the sum includes the empty set and $\mu_{\emptyset} =1$.

Now suppose $A \subset V(G), B \subset V(G)$, and let $A_L = A \cap L, A_R = A \cap R$, $B_L = B \cap L, B_R = B \cap R$.  Then, since $A$ is an independent set in $G$,  we can write
\begin{align*}
\mu_A &= \left (\frac{\lam}{1+\lam}   \right )^{|A_L| } \Pr[ A_R \subseteq \mathbf I \wedge N(A_L) \cap \mathbf I = \emptyset] \\
&=    \left (\frac{\lam}{1+\lam}   \right )^{|A_L| }  \sum_{Q \subseteq N(A_L)} (-1)^{|Q|} \mu_{Q \cup A_R}  \,.
\end{align*}
We can write similar formulae for $\mu_B$ and $\mu_{A \cup B}$.  Considering $\mu_A \mu_B$, and using the fact that $N(A_L) \cap N(B_L) = \emptyset$ since $D(A,B)>2$, we can simplify
\begin{align*}
&\left (  \sum_{Q \subseteq N(A_L)} (-1)^{|Q|} \mu_{Q \cup A_R}   \right )  \left( \sum_{Q' \subseteq N(B_L)} (-1)^{|Q'|} \mu_{Q' \cup B_R}   \right) \\
 & \hspace{6cm}=   \sum_{Q \subseteq N(A_L), Q' \subseteq N(B_L)} (-1)^{|Q \cup Q'|} \mu_{Q \cup A_R} \mu_{Q' \cup B_R} \, .
\end{align*}

Therefore we have
\begin{align*}
\mu_{A \cup B} - \mu _A \mu_B  &=   \left (\frac{\lam}{1+\lam}   \right )^{|A_L| + |B_L| }  \sum_{Q \subseteq N(A_L), Q' \subseteq N(B_L)} \left[  \mu_{Q \cup A_R \cup Q' \cup B_R} - \mu_{Q \cup A_R} \mu_{Q' \cup B_R}    \right ] \,,
\end{align*}
and so, as $\lam/(1+\lam) < 1$, 
\begin{align*}
\left| \mu_{A \cup B} - \mu _A \mu_B \right| &\le  2^{|N(A_L)| + |N(B_L)| } \max_{Q \subseteq N(A_L), Q' \subseteq N(B_L)}  \left|   \mu_{Q \cup A_R \cup Q' \cup B_R} - \mu_{Q \cup A_R} \mu_{Q' \cup B_R}  \right | \,.
\end{align*}

The first factor is bounded by a constant that depends only on $|N(A)|, |N(B)|$.   Now let $\hat A = Q \cup A_R$ and $ \hat B = Q' \cup B_R$.  Then $\hat A \subset R$ and $\hat B \subset R$, and moreover $D(\hat A, \hat B) \ge D(A,B) -2$, and so we can apply the special case above to obtain 
 \begin{align*}
\left|  \mu_{Q \cup A_R \cup Q' \cup B_R} - \mu_{Q \cup A_R} \mu_{Q' \cup B_R}    \right | &\le C' e^{ - \eta D(A,B)/2 +1} \,,
\end{align*}
and combining these bounds we get 
\begin{align*}
\left| \mu_{A \cup B} - \mu _A \mu_B \right| &\le C e^{-\eta D(A,B)/2}
\end{align*}
for a  constant $C$ that depends only on $\eta, |A|, |B|, |N(A)|, |N(B)|$. Setting $\eps = \eta/2$ we obtain Theorem~\ref{thmCorDecay}.  
\end{proof}

\section*{Acknowledgements}
The authors thank Guus Regts, Piyush Srivastava, Prasad Tetali, and Matthew Jenssen for many helpful conversations.

This work was done while the authors were participating in the Geometry of Polynomials program at the Simons Institute for the Theory of Computing.  SC is supported by NSF award DMS-1803325.  WP supported in part by NSF Career award DMS-1847451.

\end{document}